\documentclass[12pt]{article}
\usepackage{amsmath}
\usepackage{amssymb}
\usepackage{amsthm}
\usepackage{enumerate}
\usepackage[mathscr]{eucal}
\usepackage{eqlist}
\usepackage{amsfonts}
\usepackage{graphicx}
\theoremstyle{plain}
\newtheorem{theorem}{Theorem}%[section]

\newtheorem{lemma}[theorem]{Lemma}
\newtheorem{corollary}[theorem]{Corollary}
\newtheorem{conjecture}[theorem]{Conjecture}
\newtheorem{claim}[theorem]{Claim}
\theoremstyle{definition}

\theoremstyle{remark}

\numberwithin{equation}{section}
\begin{document}
\title{On the Number of Cycles in a Graph}
\author{Bader AlBdaiwi \\
%EndAName
\textit{bdaiwi@cs.ku.edu.kw }\\
Computer Science Department, Kuwait University,\\
P.O. Box 5969\\
AlSafat, 13060 \\
Kuwait\\
}
\maketitle
\begin{abstract}
There is a sizable literature on investigating the minimum and maximum
numbers of cycles in a class of graphs. However, the answer is known only
for special classes. This paper presents a result on the smallest number of
cycles in hamiltonian 3-connected cubic graphs. Further, it describes a
proof technique that could improve an upper bound of the largest number of
cycles in a hamiltonian graph.
\end{abstract}

\section{Introduction}

One of the oldest problems in graph theory that in fact goes back to the end
of the 19th century, see~\cite{A}, is the question: ``What are the smallest
and the largest number of cycles in a class of graphs?''

It turns out that it is most convenient to study the largest number of
cycles in a graph $G=(V,E)$ with respect to its cyclomatic number $%
r(G)=\left\vert E\right\vert -\left\vert V\right\vert +1.$ Let $M(r)$ be the
largest number of cycles among all graphs with the cyclomatic number $r$. In
1897 Ahrens showed that $M(r)\leq 2^{r}-1$~\cite{A}. A big step forward is
due to Entringer and Slater~\cite{Ent}. They showed that when studying the
value of $M(r)$ one can confine himself/herself to cubic graphs. Namely,
they proved that there is a cubic graph with the cyclomatic number $r$ and $%
M(r)$ cycles. Also they conjectured that $M(r)$ is asymptotically equal to $%
2^{r-1}$. There are upper bounds on $M(r)$ where the error-term is not
exponential, see for example~\cite{R} and the references given there.
Aldred and Thomassen~\cite{Aldred1} proved that $M(r)\leq \frac{15}{16}%
2^{r}+o(2^{r})$. This is still the only upper bound that improves on the
coefficient of the leading term. As to the lower bounds, the best one so far
was published recently in~\cite{H}, where the error term is exponential.

As for the smallest number of cycles, it has been shown in~\cite{Clark}
that a 2-connected cubic graph of $n$ vertices contains at least $%
(n^{2}+14n)/8$ cycles, and that the bound is best possible. In the same
paper, it was conjectured that the difference between the 2-connected and
3-connected cubic graphs, in a sense, is dramatic. More precisely, it was
conjectured that $f(n)$, the minimum number of cycles in a 3-connected cubic
graph, is superpolynomial. The conjecture was proved in~\cite{Aldred}. It
is shown there that, for $n$ sufficiently large, $%
2^{n^{0.17}}<f(n)<2^{n^{0.95}}$. In the same paper, it is suggested that
replacing the condition 3-connected by the condition cyclically
4-edge-connected may increase the growth of cycles' number to be exponential
in terms of $n$.
R. Aldred has conjectured (unpublished) that restricting the graphs to be
cubic hamiltonian 3-connected might lead to the same property.
The first main result of this paper supports Aldred's conjecture. We
conjecture that a graph $H_{2n},$ defined in this paper, has the smallest
number of cycles among all cubic hamiltonian 3-connected graphs. Then we
show that the number of cycles in $H_{2n}$ grows faster than Fibonacci
sequence. The second result provides a proof technique whose refinement
could lead to an improvement of the upper bound on the largest number of cycles in a
graph.

%
% Cyclomatic Complexity
%
Our research has been motivated by a computer science application.
Cyclomatic complexity is a software metric used to quantify the structure of
a computer program. A program source code can be modeled by a Flow Control
Graph~(FCG) in which nodes and edges represent code blocks and the possible
execution paths among them~\cite{mccabe1}. The cyclomatic complexity is
based on the cyclomatic number of a program FCG. A software of high
cyclomatic complexity indicates a large number of possible execution paths.
Such a software would be difficult to test and expensive to maintain.
As pointed out in~\cite{mccabe2} and~\cite{watson}, one should estimate
cyclomatic complexity in advance during the software design to avoid a
complex code structure. Therefore, investigating the number of cycles in
graphs could help developers and automated code generators in avoiding
structures that may lead to high cyclomatic complexity. It also could help
in defining software design templates that lead to building low
cyclomatic complexity software.

\section{Preliminaries}

In what follows it
is assumed that $G=(V,E)$ is a cubic hamiltonian 3-connected graph, where $%
V=\{v_{0},...,v_{n-1}\},H=v_{0}v_{1}v_{2}...v_{n-1}v_{0}$ is a hamiltonian cycle
of $G$. The set of edges in $G$ but not in $H$ will be denoted by $S$, and
called spokes. We will say that a set of spokes~$F \subseteq S$ forms a cycle if there
is a cycle~$C$ in~$G$ so that all spokes of $F$ belong to $C$ and no other
spoke is in $C.$ Such cycle $C$ will be called an F-cycle.
$v_{i}-v_{j}$ denotes the path $v_{i}v_{i+1}...v_{j}$,
which is a part of $H,$ indices taken modulo $n$. $E(T)$ stand for the edge
set of a graph $T$. The fact that a vertex $v_{k}$ is an internal vertex of
the path $v_{i}-v_{j}$ will be denoted by $v_{i}\prec v_{k}\prec v_{j}.$ The
expression $v_{i}\prec v_{k}\prec v_{j}\prec v_{m}$ is a shorthand for $%
v_{i}\prec v_{k}\prec v_{j}$ and $v_{k}\prec v_{j}\prec v_{m}.$ The next
lemma constitutes a simple but useful observation.

\begin{lemma}
\label{1} Let $F$ be a non-empty set of spokes, $\left\vert F\right\vert =k,$
and let the spokes of $F$ be incident with vertices $%
v_{i_{1}},...,v_{i_{2k}},$ $i_{1}<...<i_{2k}.$ If $F$ forms a cycle $C$ in $%
G,$ then $E(C)= E_1$ or $E(C) = E_2$,
%E_{i},$ $i\in \{1,2\},$
where

$E_{1}=F\cup \bigcup\limits_{j=1}^{k}E(v_{i_{2j-1}}-v_{i_{2j}}),$ and

$E_{2}=F\cup \bigcup\limits_{j=1}^{k-1}E(v_{i_{2j}}-v_{i_{2j+1}})\cup E(v_{i_{2k}}-v_{i_{1}}).$

In particular, $F$ forms at most two cycles in $G.$
\end{lemma}

\begin{proof}
Let $F$ be a non-empty set of spokes.
It is worth noting that no vertex is incident with both an edge in $E_1-F$ and an edge in $E_2-F$.
Clearly, both $E_{1}$ and $E_{2}$ induce 2-regular graphs. The set $F $
forms a cycle in $G$ if $E_{1}$ or $E_2$ induces a cycle.

It will be shown that there are at
most two ways how to choose the set $H^{\prime }$ of edges from the
hamiltonian cycle $H$ so that $F \cup H^{\prime }$ forms a cycle of $G.$
Suppose that $v_{i},$ $v_{j},$ and $v_{m},$ are vertices of $G$, $v_{i}\prec
v_{j}\prec v_{m}$ so that each of the three vertices is incident to a spoke
of $F,$ and for all $k,i<k<j,$ and $j<k<m,$ the vertex $v_{k}$ is not
incident to a spoke of $F.$ Then either all the edges of the path $%
v_{i}-v_{j}$ belong to $H^{\prime }$ and no edge of the path $v_{j}-v_{m}$
is in $H^{\prime }$ or no edge of the path $v_{i}-v_{j}$ is in $H^{\prime }$
and the path $v_{j}-v_{k}$ is in $H^{\prime }.$ \ In general, if $%
v_{i_{1}},...,v_{i_{2k}},i_{1}<...<i_{2k},$ were vertices of $G$ incident to
a spoke of $F,$ then either all the edges of paths $%
v_{i_{1}}-v_{i_{2}},v_{i_{3}}-v_{i_{4}},...,v_{i_{2k-1}}-v_{i_{2k}},$ would
be in $H^{\prime },$ or all the edges of the paths $%
v_{i_{2}}-v_{i_{3}},v_{i_{4}}-v_{i_{5}},...,v_{i_{2k}}-v_{i_{1}},$ would be
in $H^{\prime }.$ Thus, if $F$ forms a cycle in $G$, then the edge set of $C$
is either $E_{1}$ or $E_{2}$, and consequently, $F$ forms at most two cycles
in $G$.
\end{proof}

\section{Smallest Number of Cycles}

Let $H_{2n}=(V,E)$ be a cubic hamiltonian 3-connected graph
where ${n\geq 2}$, $V=\{v_{0}, v_{1}, ...,v_{2n-1}\},$ and $E$ comprises a
hamiltonian cycle $C=v_{0},v_{1},...,v_{2n-1},v_{0}$ and a set of spokes
$S_{n},$ $\left\vert S_{n}\right\vert = n$ such that:
\begin{eqnarray*}
e_{0} &=&v_{1}v_{2n-1}\in S_{n} \\
e_{i} &=&
          \begin{cases}
            v_{i-1}v_{2n-i-2}\in S_{n},\text{ for }\  1 \le i < n-1,\ i \text{ is odd} \\%\ i \equiv 1\ (\func{mod}2)\\
            v_{i+1}v_{2n-i}\in S_{n},\text{ for }\  2 \le i < n-1,\ i \text{ is even}\\%\equiv 0\ (\func{mod}2)\\
          \end{cases}
          \\ \\
e_{n-1} &=&
\begin{cases}
            v_{n-1}v_{n+1}\in S_{n},\ n \text{ is odd}\\%\text{ for }\  n \equiv 1\ (\func{mod}2)\\
            v_{n-2}v_{n}\in S_{n},\ n \text{ is even}\\%\text{ for }\  n\equiv 0\ (\func{mod}2)\\
\end{cases}
\end{eqnarray*}

See Figure~\ref{fig1} and Figure~\ref{fig2} for $H_{16}$ and $H_{18}.$

\begin{figure}[h]
\begin{center}
\includegraphics[width=0.8\textwidth]{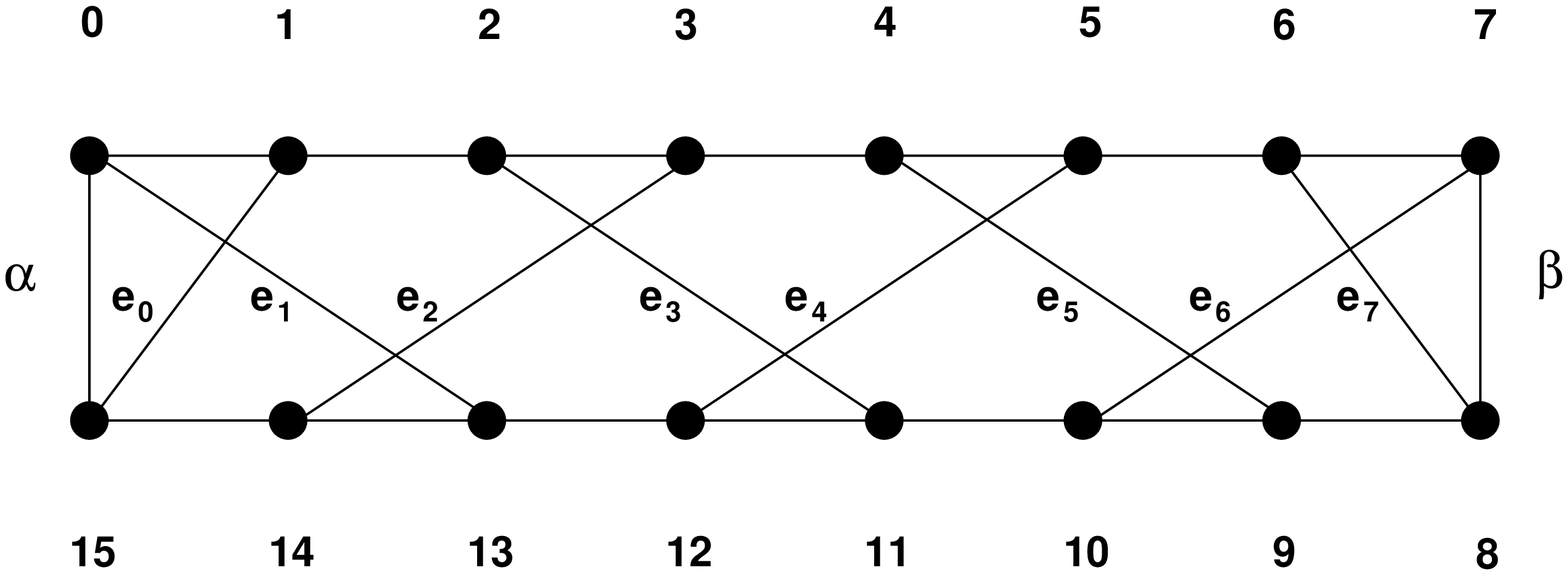}
\end{center}
\caption{$H_{16}$}
\label{fig1}
\end{figure}
\begin{figure}[h]
\begin{center}
\includegraphics[width=0.8\textwidth,
height=0.22\textheight]{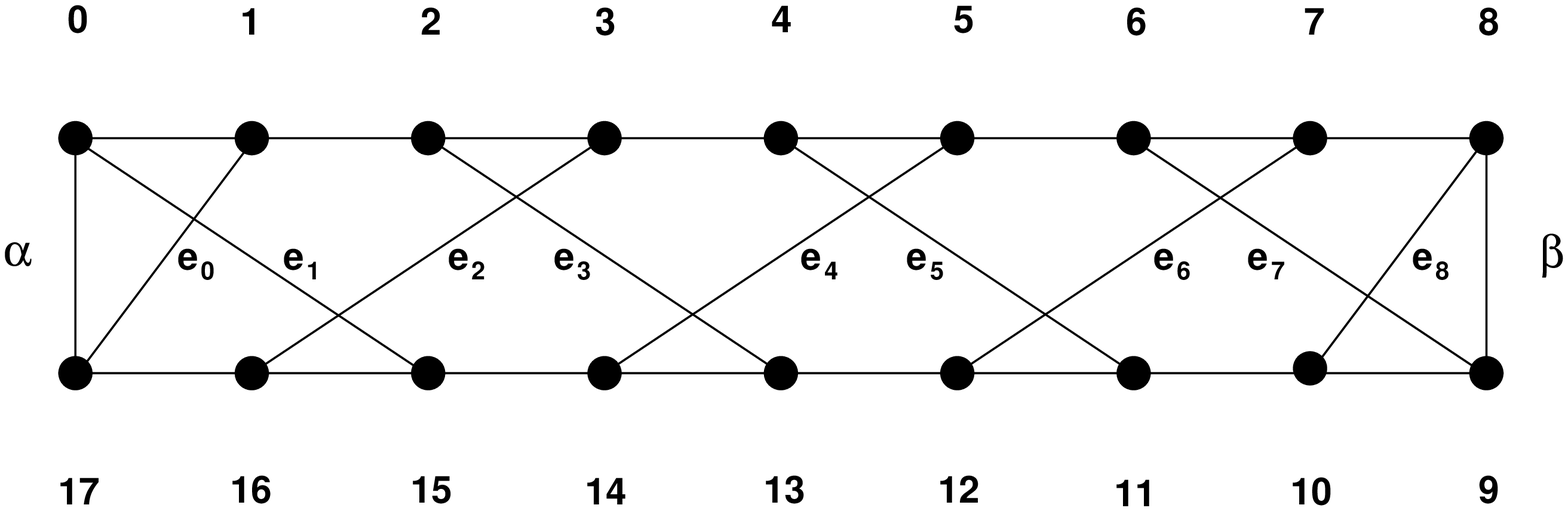}
\end{center}
\caption{$H_{18}$}
\label{fig2}
\end{figure}

We believe that:

\begin{conjecture}
\label{CC}The graph $H_{2n}$ has the smallest number of cycles among all
hamiltonian cubic 3-connected graphs on $n$ vertices.
\end{conjecture}

Let $I(n)\,=(S_{n},E^{'})$ denote the graph where the vertex set is the set of
spokes of $H_{2n}$ and two vertices are adjacent if the corresponding spokes
intersect. Since $H_{2n}$ is 3-connected, $I(n)$ is a connected graph. $I(n)$
being a path is one reason to believe that Conjecture~\ref{CC} is true. Note
that $I(n)$ is a path with the first vertex $e_{0}\,$ and the last $e_{n-1}.$

Let $\alpha$ be the edge $(v_0, v_{2n-1})$, and $\beta$ be the edge $(v_{n-1}, v_n)$.
The following claim states a very important property of the graph $H_{2n}.$

\begin{claim}
Let $H$ be a graph obtained by removing an edge $e\in \{\alpha ,\beta
,e_{0},e_{n-1}\}$ from $H_{2n}$ and suppressing two vertices of degree 2.
Then $H$ is isomorphic to $H_{2n-2}.$
\end{claim}

\begin{proof}
The situation after removing the edge $\alpha$ or $e_0$ is clear from
Figure~\ref{fig3} and Figure~\ref{fig4}, respectively.
The situation after deleting the edge $\beta$ or $e_{n-1}$ is analogous.
%It is clear from Figure~\ref{fig3} and Figure~\ref{fig4}.
\end{proof}

\begin{figure}[h]
\begin{center}
\includegraphics[width=0.8\textwidth]{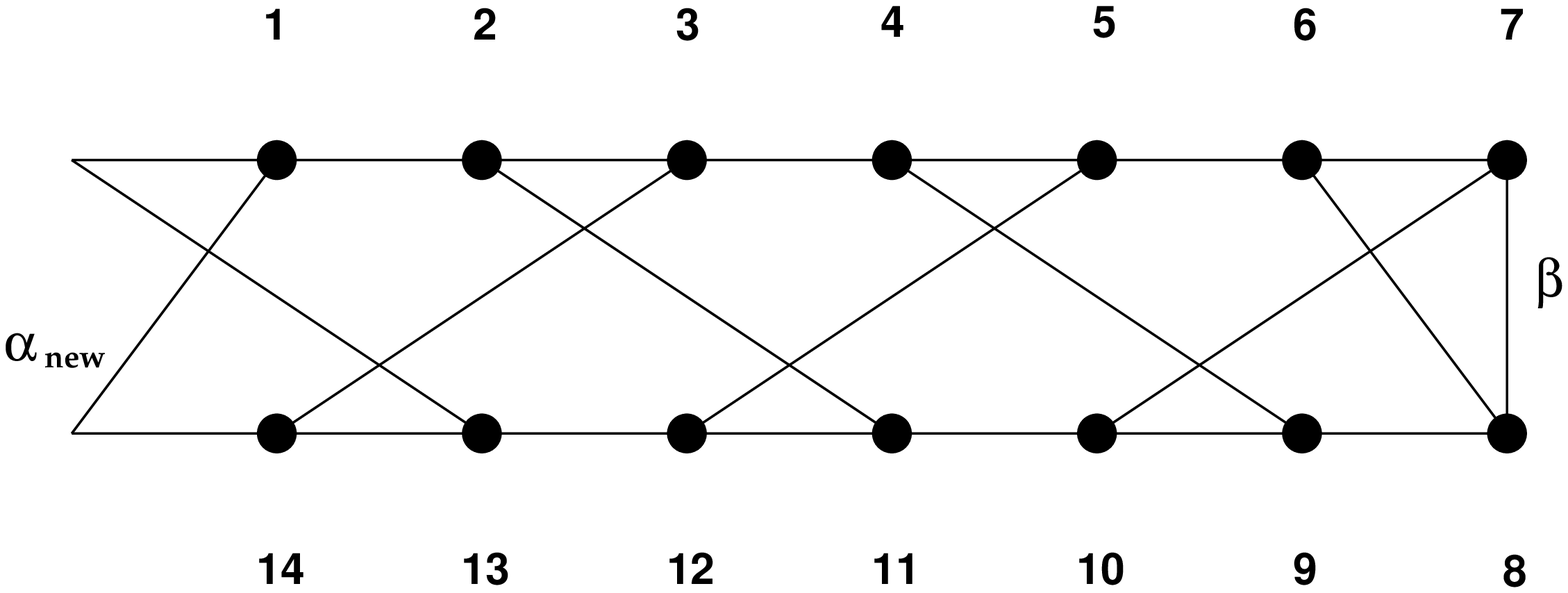}
\end{center}
\caption{$H_{16} - \protect\alpha \simeq H_{14}$}
\label{fig3}
\end{figure}
\begin{figure}[h]
\begin{center}
\includegraphics[width=0.8\textwidth,
height=0.22\textheight]{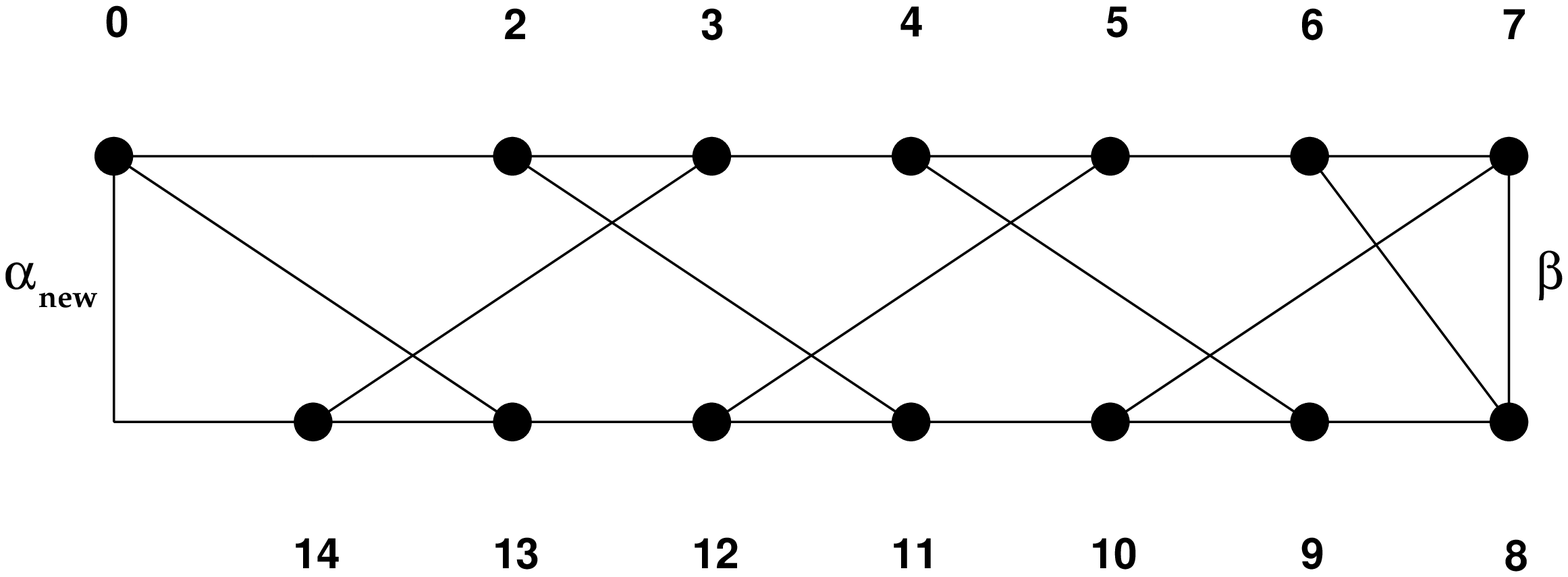}
\end{center}
\caption{$H_{16} - e_0 \simeq H_{14}$}
\label{fig4}
\end{figure}

Let $F \subseteq S_{n}$ be a set of spokes.
The Basic Interval Representation of $F,$ the BIR of $F,$ is a
partition of $F$ into minimal number of parts $F_{1},...,F_{k}$ such that
(i) For each $i,i=1,2,...,k,$ $F_{i}$ is a set of spokes so that the indices
of spokes in $F_{i\text{ }}$ are consecutive numbers, that is, they form an
interval; (ii) If $e_{s}\in F_{i}\,,e_{t}\in F_{j},$ and $i<j,$ then $s<t.$

Clearly, for each set $F$ its BIR is determined in a unique way. The part $%
F_{k\text{ }}$ will be called the last part of BIR of $F.$

If a set of spokes comprises spokes so that their indices form an interval,
then for short $F$ will be called a set of consecutive spokes.

\begin{lemma}
\label{a} Let $F$ be a set of consecutive spokes. Then there are two $F$%
-cycles in $H_{2n}.$ Further, if $\left\vert F\right\vert $ is even, then one
of the two cycles contains both edges $\alpha $ and $\beta ,$ while the
other cycle contains neither $\alpha $ nor $\beta .$ For $\left\vert
F\right\vert $ odd, one of the two cycles contains $\alpha $ and not $\beta
, $ the other contains $\beta $ but not $\alpha .$
\end{lemma}

\begin{proof}
The statement follows directly from Lemma~\ref{1}. It is easy to check in
this case that both $E_1$ and $E_2$ induce a single cycle.
\end{proof}

As a direct consequence:

\begin{corollary}
\label{b} Let $F_{1},...,F_{k},k\geq 2,$ be BIR of a set $F$ of spokes of $%
H_{2n}$. Then there is at most one $F$-cycle in $H_{2n}.$ An $F$-cycle $C$
exists iff $\left\vert F_{i}\right\vert $ is even for all $i=2,...,k-1.$
Further, $C$ contains the edge $\alpha $ (the edge $\beta $) if and only if
$\left\vert F_{1}\right\vert $ is even ($\left\vert F_{k}\right\vert $ is
even).
\end{corollary}

\begin{proof}
It follows directly from the above lemma.
\end{proof}

Let $c_{n}$ be the total number of cycles in $H_{2n},$ and let $%
\alpha_{n}$ stand for the number of cycles in $H_{2n}$ containing the edge
$\alpha.$ By inspection, using Corollary~\ref{b}, we get, $%
c_{2}=7,c_{3}=14,c_{4}=26,c_{5}=46,$ and $\alpha _{2}=4,$ $\alpha _{3}=7,$ $%
\alpha _{4}=12,$ and $\alpha_{5}=20.$ Further, let $E_{n}$ and $O_{n}$ be
the number of cycles $C$ in $H_{2n}$ containing the edge $\alpha \,$ so that
if $F_{C}$ is the set of all spokes of $C,$ then the last part in the BIR of
$F_{C}$ is of even, or odd parity, respectively. For short, these cycles
will be called $\alpha \mathit{-even}$ and $\alpha \mathit{-odd}$ cycles,
respectively. By definition,
\[
\alpha_{n}=E_{n}+O_{n}
\]

Clearly, if $F$ is a set of spokes of $H_{2n}$ not containing the last
spoke $e_{n-1}$, then there is an $F$-cycle in $H_{2n}$ iff there is an
$F^{\prime}$-cycle in $H_{2n-2}$.

\begin{lemma}
\label{c}For $n\geq 4,$ (i) $O_{n}=\alpha_{n-1};$ (ii) $E_{n}=\alpha_{n-1}- O_{n-2}.$
\end{lemma}

\begin{proof}
Let $C$ be an $\alpha \mathit{-odd}$ cycle in $H_{2n}.$ If $F_{C},$ the set
of spokes of $C,$ contains the last spoke $e_{n-1},$ then the $F_{C}-e_{n-1}$
cycle is an $\alpha \mathit{-even}$ cycle of $H_{2n-2},$ otherwise $F_{C}$-cycle
forms an $\alpha \mathit{-odd}$ cycle of $H_{2n-2}.$
Thus, $O_{n}=E_{n-1}+O_{n-1}= \alpha_{n-1}$ and (i) follows.
For $C$ being an $\alpha \mathit{-even}$ cycle,
if $F_{C}$ does not contain the last spoke $e_{n-1}$,
then $F_{C}$-cycle forms an $\alpha \mathit{-even}$ cycle of $H_{2n-2},$ otherwise $F_{C}-e_{n-1}$
cycle forms an $\alpha \mathit{-odd}$ cycle of $H_{2n-2}$ that contains $e_{n-1}$ the last
spoke of $H_{2n-2}.$
Hence,
$E_n = E_{n-1} +E_{n-2} =(\alpha_{n-1} - O_{n-1})+(\alpha_{n-2} - O_{n-2})= \alpha_{n-1} - O_{n-2}$, since $O_{n-1}=\alpha_{n-2}$.
\end{proof}

\begin{corollary}
\label{d}For $n\geq 4,$ $\alpha _{n}=\alpha _{n-1}+\alpha _{n-2}+1.$
\end{corollary}

\begin{proof}
By Lemma \ref{c}, $\alpha _{n}=E_{n}+O_{n}=2\alpha _{n-1}-O_{n-2}=2\alpha
_{n-1}-\alpha _{n-3}.$ We show by induction that, for $n\geq 4,$ $\alpha
_{n}=\alpha _{n-1}+\alpha _{n-2}+1.$ It is $\alpha _{4}=12=7+4+1.$ By the
induction hypothesis, for $n>4,$ $\alpha _{n-1}-\alpha _{n-3}=\alpha
_{n-2}+1.$ Therefore, $\alpha _{n}=2\alpha _{n-1}-\alpha _{n-3}=\alpha
_{n-1}+(\alpha _{n-1}-\alpha _{n-3})=\alpha _{n-1}+\alpha _{n-2}+1.$
\end{proof}

Now it is possible to proceed with counting the number of cycles in $H_{2n}.$

\begin{theorem}
The number of cycles of $H_{2n}$ is:
\end{theorem}

\[
c_{n}=\left( \frac{2}{\sqrt{5}}+1\right) \left( \frac{1+\sqrt{5}}{2}\right)
^{n+2}+\left( 1-\frac{2}{\sqrt{5}}\right) \left( \frac{1-\sqrt{5}}{2}\right)
^{n+2}-(n+4).
\]

\begin{proof}
Since the number of cycles of $H_{2n}$ not containing the edge $\alpha $
equals the number of cycles of $H_{2n-2\text{ }}$, then:
\[
c_{n}=c_{n-1}+\alpha _{n}
\]%
To determine $c_{n}$ we first prove:

\[
c_{n}=\alpha _{n+2}-(n+3)\text{ for }n\geq 2.
\]%
Using the initial values given above, the formula can be verified to be true
for $n=2.$ For $n>2,$ by the induction hypothesis and Corollary \ref{d}, it
is $c_{n}=c_{n-1}+\alpha _{n}=\alpha _{n+1}-(n+2)+\alpha _{n}=\alpha
_{n+2}-(n+3).$

Solving the recurrence relation $\alpha _{n}=\alpha _{n-1}+\alpha _{n-2}+1,$
with the initial conditions $\alpha _{2}=4,\alpha _{3}=7,$ one gets:

\[
\alpha _{n}=\left( \frac{2}{\sqrt{5}}+1\right) \left( \frac{1+\sqrt{5}}{2}%
\right) ^{n}+\left( 1-\frac{2}{\sqrt{5}}\right) \left( \frac{1-\sqrt{5}}{2}%
\right) ^{n}-1
\]

which in turn implies:

\[
c_{n}=\left( \frac{2}{\sqrt{5}}+1\right) \left( \frac{1+\sqrt{5}}{2}\right)^{n+2}+
\left( 1-\frac{2}{\sqrt{5}} \right) \left( \frac{1-\sqrt{5}}{2}\right)^{n+2}-(n-4).
\]
\end{proof}

\textbf{Remark.} The ``pseudo'' Fibonacci sequence $\{c_{n}\}$ is
identical to sequence A001924 electronically published at The On-Line
Encyclopedia of Integer Sequences (OEIS) - https://oeis.org/A001924. The
same sequence appears also in page 58 of~\cite{Leon}, and in~\cite{Anton}.

\section{Largest Number of Cycles}

It is widely believed, but not proved, that the largest number of cycles
among all cubic graphs is attained at a hamiltonian graph. In this section it
is shown that the largest number of cycles among all hamiltonian cubic
graphs on $n$ vertices, denoted by $T(n),$ can be upper bounded by $%
T(n)\leq 2^{\frac{n}{2}+1}-f(n),$ where $f$ is an exponential function. We
note that the cyclomatic number $r$ of a cubic graph on $n$ vertices equals $%
\frac{n}{2}+1.$ This result is superseded by Aldred and Thomasen in~\cite{Aldred1}
who proved that $M(r)\leq \frac{15}{16}2^{r}+o(2^{r}).$ However, we believe that
a refinement of the method used here could lead to further improvement of
the above mentioned result.

\begin{theorem}
\label{3}$T(n)<2^{\frac{n}{2}+1}-2^{\frac{n}{2}-2\sqrt{n}-3}.$
\end{theorem}

Before proving the statement, some more notions are introduced. Two spokes $%
e=v_{i}v_{j},f=v_{k}v_{m}$ are intersecting if the path $v_{i}-v_{j}$
contains one of the two vertices $v_{k},v_{m},$ and the path $v_{j}-v_{i}$
contains the other of the two vertices, otherwise it is said that they are
parallel. Further, let $F$ be a set of spokes, $e,f\in F.$ Then, $e$ and $f$
are consecutive spokes in $F$, if each $g\in F$, $e\neq g\neq f,$ either
intersects both $e$ and $f$, or $g$ is parallel to both $e$ and $f$. It is
easy to see that if $e,f$ are two spokes of $F$ incident with vertices $%
v_{i},v_{j},v_{k},v_{m}$, where, say $v_{i}\prec v_{j}\prec v_{k}\prec v_{m}$%
, then $e$ and $f$ are consecutive iff either no internal vertex of the
paths $v_{i}-v_{j}$ and $v_{k}-v_{m}$ or no internal vertex of the paths $%
v_{j}-v_{k}$ and $v_{m}-v_{i}$ is incident with a spoke in $F$.

\begin{proof}
We start with a series of claims.

\textbf{Claim~\ref{3}.1.} Let $F$ be a set of spokes so that $e,f$ be
consecutive spokes in $F$, and the number of spokes of $F$ intersecting both
$e$ and $f$ be even. Then $F$ forms at most one cycle in $G$.

\textit{Proof of Claim~\ref{3}.1.} The statement is immediate for $|F|=2$.
Assume now $|F|\geq 3.$ Let spokes $e,f$ be incident to vertices $%
v_{i},v_{j},v_{k},v_{m} $, where, $v_{i}\prec v_{j}\prec v_{k}\prec v_{m}$.
In addition, we suppose WLOG that no internal vertex of the paths $%
v_{i}-v_{j}$ and $v_{k}-v_{m}$ is incident with a spoke in $F$. Since the
number of spokes in $F$ intersecting $e$ and $f$ is even, hence, we have
that if a cycle $C$ formed by the spokes of $F$ contains all edges of the
path $v_{i}-v_{j}$, then $C$ would have to contain also all the edges of the
path $v_{k}-v_{m}$. However, then the two path together with $e$ and $f$
form a cycle. Therefore, at least one of $E_{1},E_{2}$ defined in Lemma \ref%
{1} does not induce a single cycle. Thus, $F$ forms at most one cycle in $G$.

\textbf{Claim~\ref{3}.2.} Let $F,|F|>1,$ be a set of spokes, $e\in F$ be so
that no spoke in $F$ intersects $e$. Then $F$ forms at most one cycle in $G$.

\textit{Proof of Claim~\ref{3}.2.} Let $e=v_{i}v_{j}$. Suppose first that
both $v_{i}-v_{j}$ path and $v_{j}-v_{i}$ path contain an internal vertex
incident to a spoke in $F$. Then clearly $F$ does not form any cycle in $G$.
Otherwise, let no internal vertex of the path $P=v_{i}-v_{j}$ be incident to
a spoke in $F$. Then $e$ and the path $P$ form a cycle, hence at least one
of $E_{1},E_{2}$ does not induce a single cycle.

\textbf{Claim~\ref{3}.3.} There is in $G$ a set $F$ of spokes, $|F|\geq
\frac{n}{2}-2\left\lfloor \sqrt{n}\right\rfloor $ so that either $F$
contains a pair of consecutive spokes, or $F$ contains a spoke that is
intersected by no spoke in $F$.

\textit{Proof of Claim~\ref{3}.3.} First, let $P$ be a partition of the set
$0,1,...,n-1$ into $k$ parts $P_{1},...,P_{k}$, so that $P_{i}$ contains a
set of consecutive integers, and $k\leq |P_{i}|\leq k+2$ for $i=1,...,k$. To
see that such partition is possible, set $k=\left\lfloor \sqrt{n}%
\right\rfloor $. Choose $t$ so that $(t-1)^{2}\leq n<t^{2}$. Then $k=t-1$,
and the existence of $P$ follows from $k(k+2)=(t-1)(t+1)=t^{2}-1\geq n$.

Suppose first that there is a spoke $e=v_{i}v_{j},\in G,i<j$ so that both $%
i,j\in P_{m}$, for some $1\leq m\leq k$. To get a set $F$ with the required
properties it suffices to remove from $S,$ the set of all spokes, those
spokes that are incident to internal vertices of the path $v_{i}-v_{j}$. As
at most $k$ spokes can be removed, and $e$ is not intersected by any spoke
in $F$, the proof follows.

In the other case, let $T$ be the set of spokes incident to vertices in $%
P_{1}$. Since no spoke in $T$ is incident to two vertices in $P_{1}$, by the
pigeon hole principle, there is an $m$, $1<m\leq k$, so that there are at
least two spokes in $T$, say $e=v_{i}v_{a}$ and $f=v_{j}v_{b},$ \thinspace $%
i,j\in P_{1}$ so that $a,b$ belong to $P_{m}.$ Let $i<j,a<b.$ To get a
requested set $F$ of spokes, it suffices to remove from $S$ all spokes
incident to an internal vertex of $v_{i}-v_{j}$, and all spokes incident to
an internal vertex of $v_{a}-v_{b}$. As $|P_{t}|\leq k+2$ for all $t=1,...,k$%
, at most $2k$ spokes can be removed. Thus, $|F|\geq \frac{n}{2}-2k$.

Now we are ready to prove the theorem. Let $F$ be a set of spokes guaranteed
by Claim~\ref{3}.3. Suppose that $F$ contains two consecutive spokes $e$ and $f$,
and that there are $t$ spokes in $F$ intersecting both $e$ and $f$. Then
there are $\frac{1}{2}2^{t}2^{|F|-t-2}=2^{|F|-3}$ subsets of $F$ satisfying
the assumption of Claim~\ref{3}.1, and the statement follows. In the case that $F$
contains a spoke not intersected by any spoke in $F$, then there are $%
2^{|F|-1}$ subsets of $F$ fulfilling the assumptions of Claim~\ref{3}.2. The proof
is complete.
\end{proof}

\bigskip

\section{Computational Results}

%We believe that it will be a very difficult task to determine values of the
%function $T(2n)$.
To provide supporting evidence for Conjecture~\ref{CC} on the smallest number of cycles,
we used an extensive computer search to verify it for all $H_{2n}$, $2n \le 16$. It turns out
that there are several extremal graphs for small values of $2n$ vertices as shown
in Table~\ref{Table1}.
\begin{table}[ht]
\centering % used for centering table
\begin{tabular}{c c} % centered columns (2 columns)
\hline\hline %inserts double horizontal lines
2n & Number of Graphs \\ [0.5ex] % inserts table
%heading
\hline
\hline % inserts single horizontal line
6 & 1 \\
8 & 2 \\
10 & 2 \\
12 & 5 \\
14 & 7 \\
16 & 14 \\ [1ex]
\hline
\hline %inserts single line
\\
\end{tabular}
\caption{Number of Extremal Graphs Exhibiting The Smallest Number of Cycles Including $H_{2n}$} % Table Title
\label{Table1} % is used to refer this table in the text
\end{table}

\bigskip
Figures~\ref{Extremal 8}~to~\ref{Extremal 14} show the extremal graphs for $2n=8$ to $2n=14$ excluding $H_{2n}$.
\begin{figure}[H]
\begin{center}
\includegraphics[width=0.27\textwidth]{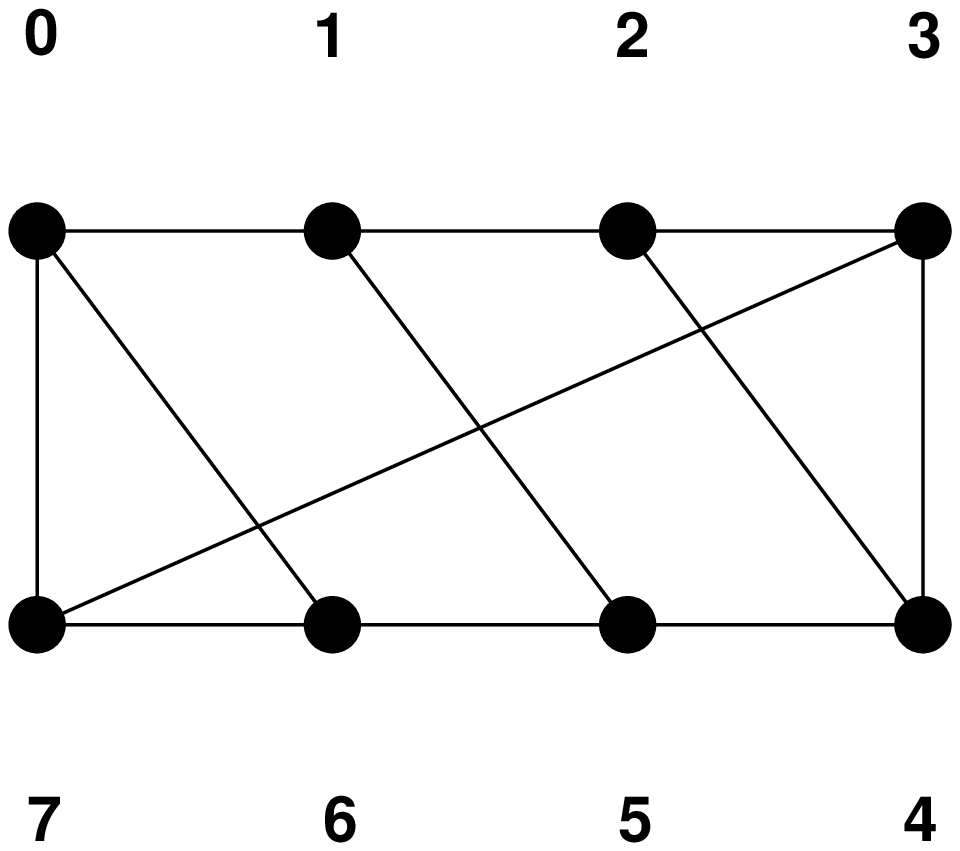}
\end{center}
\caption{The extremal graph other than $H_8$ with the smallest number of cycles for $2n=8$}
\label{Extremal 8}
\end{figure}
\begin{figure}[H]
\begin{center}
\includegraphics[width=0.33\textwidth]{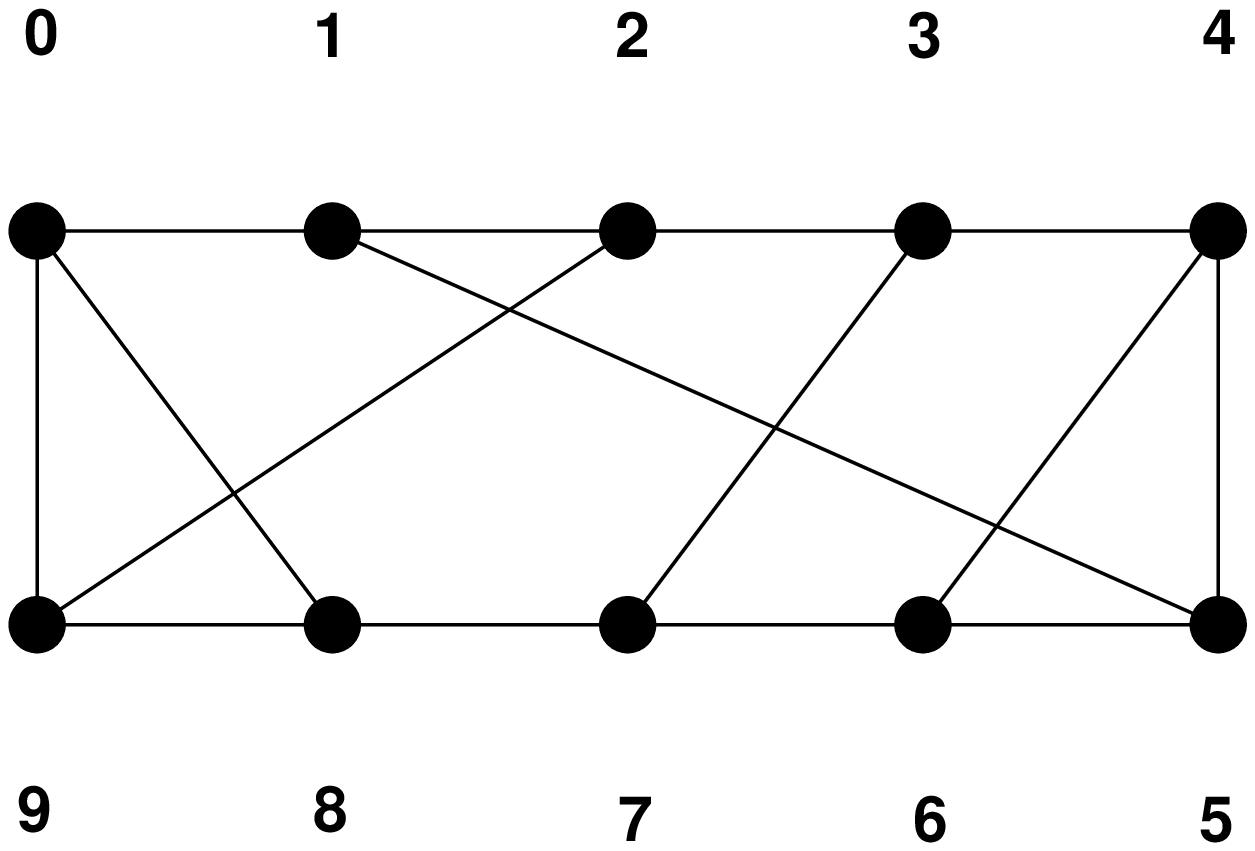}
\end{center}
\caption{The extremal graph other than $H_{10}$ with the smallest number of cycles for $2n=10$}
\label{Extremal 10}
\end{figure}
\begin{figure}[H]
\begin{center}
\includegraphics[width=0.4\textwidth]{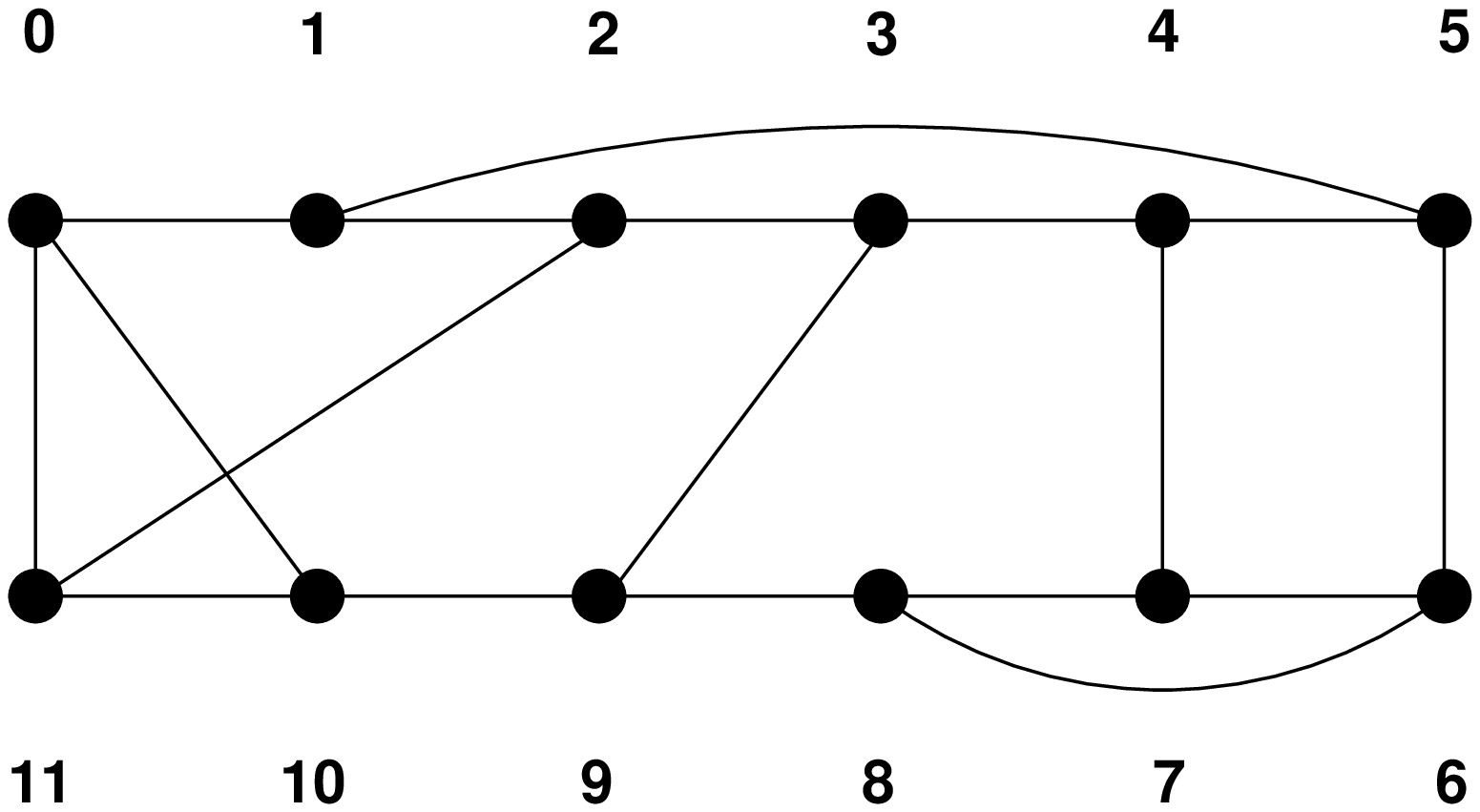}
\hspace{1 cm}
\includegraphics[width=0.4\textwidth]{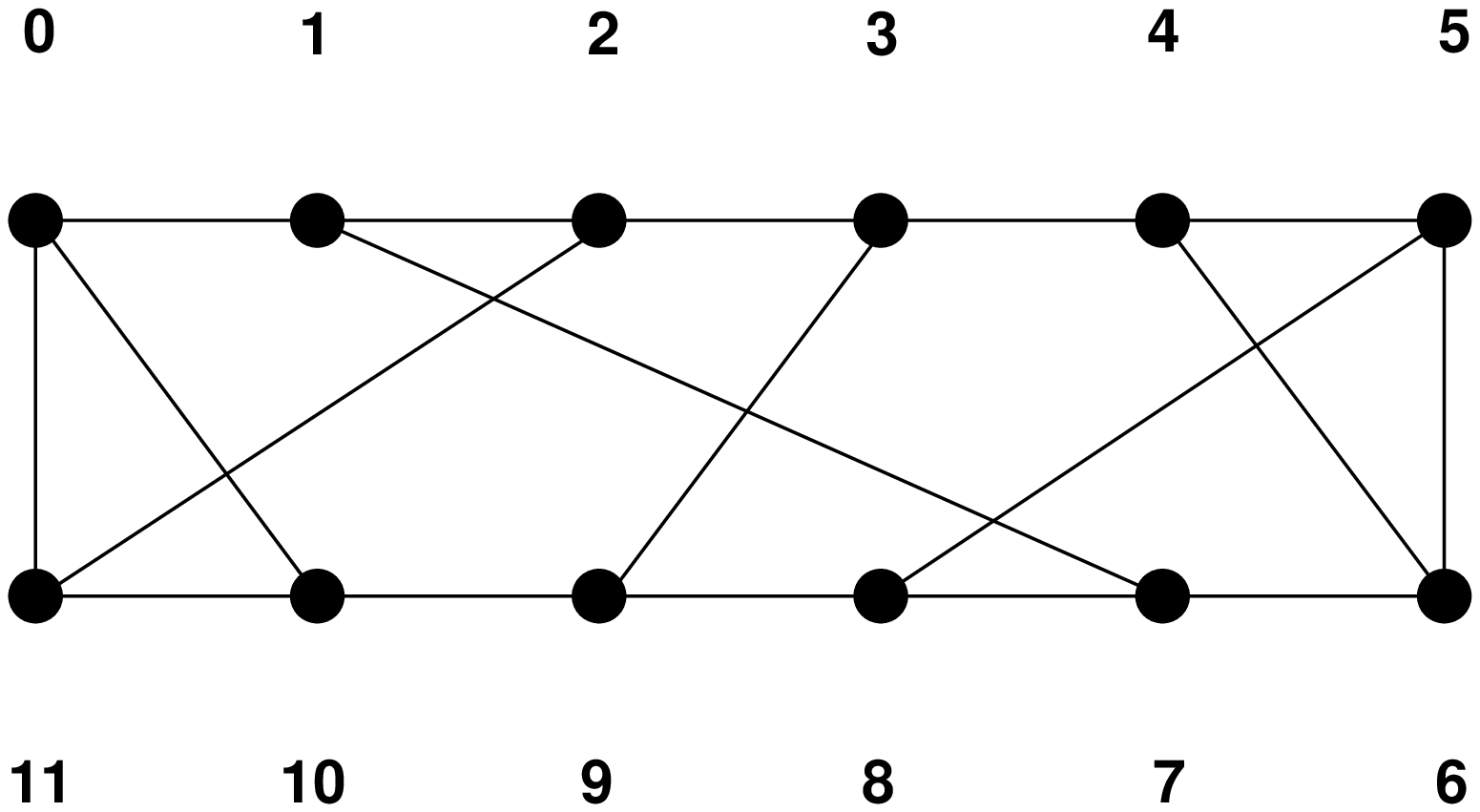}

\bigskip
\bigskip
\includegraphics[width=0.4\textwidth]{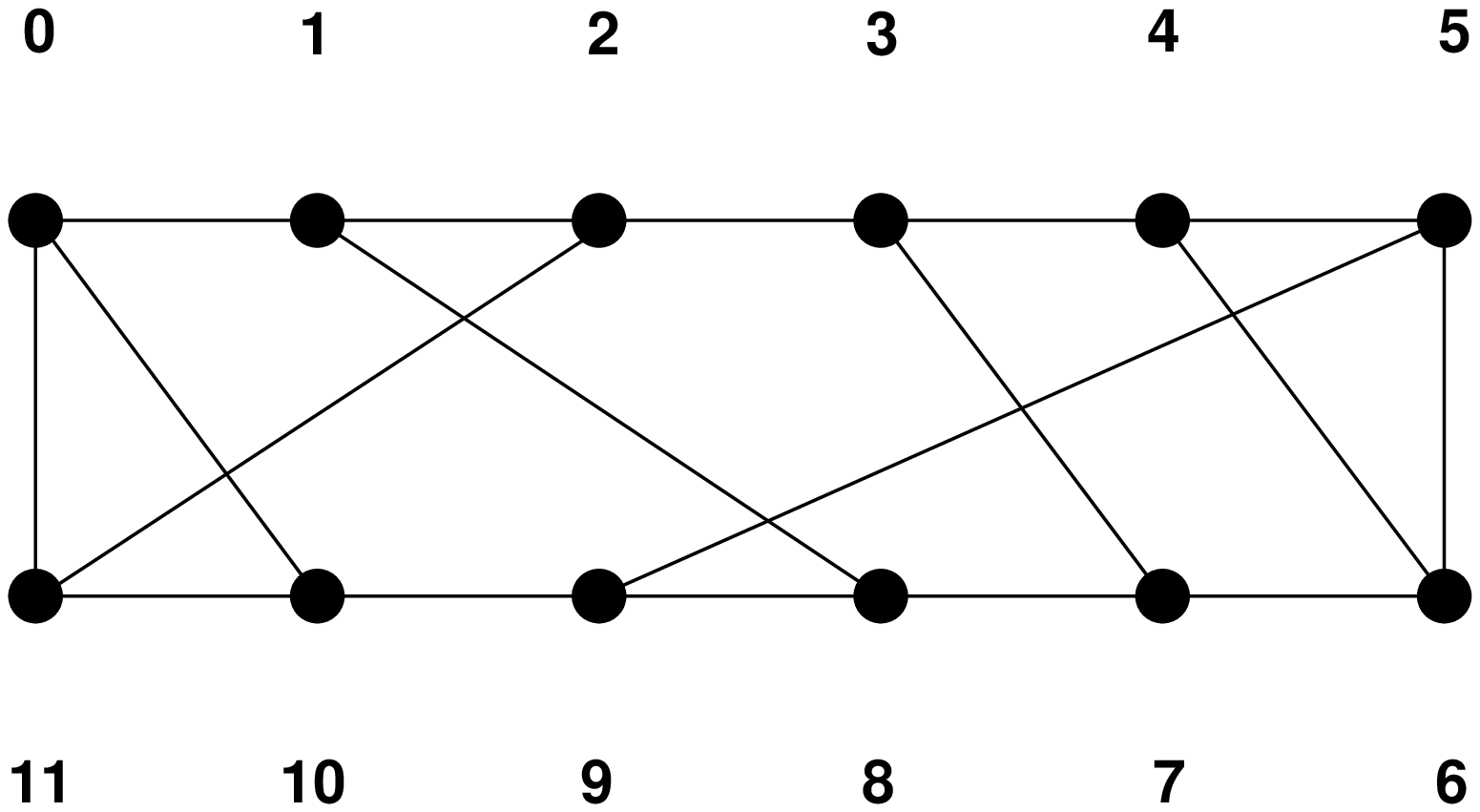}
\hspace{1 cm}
\includegraphics[width=0.4\textwidth]{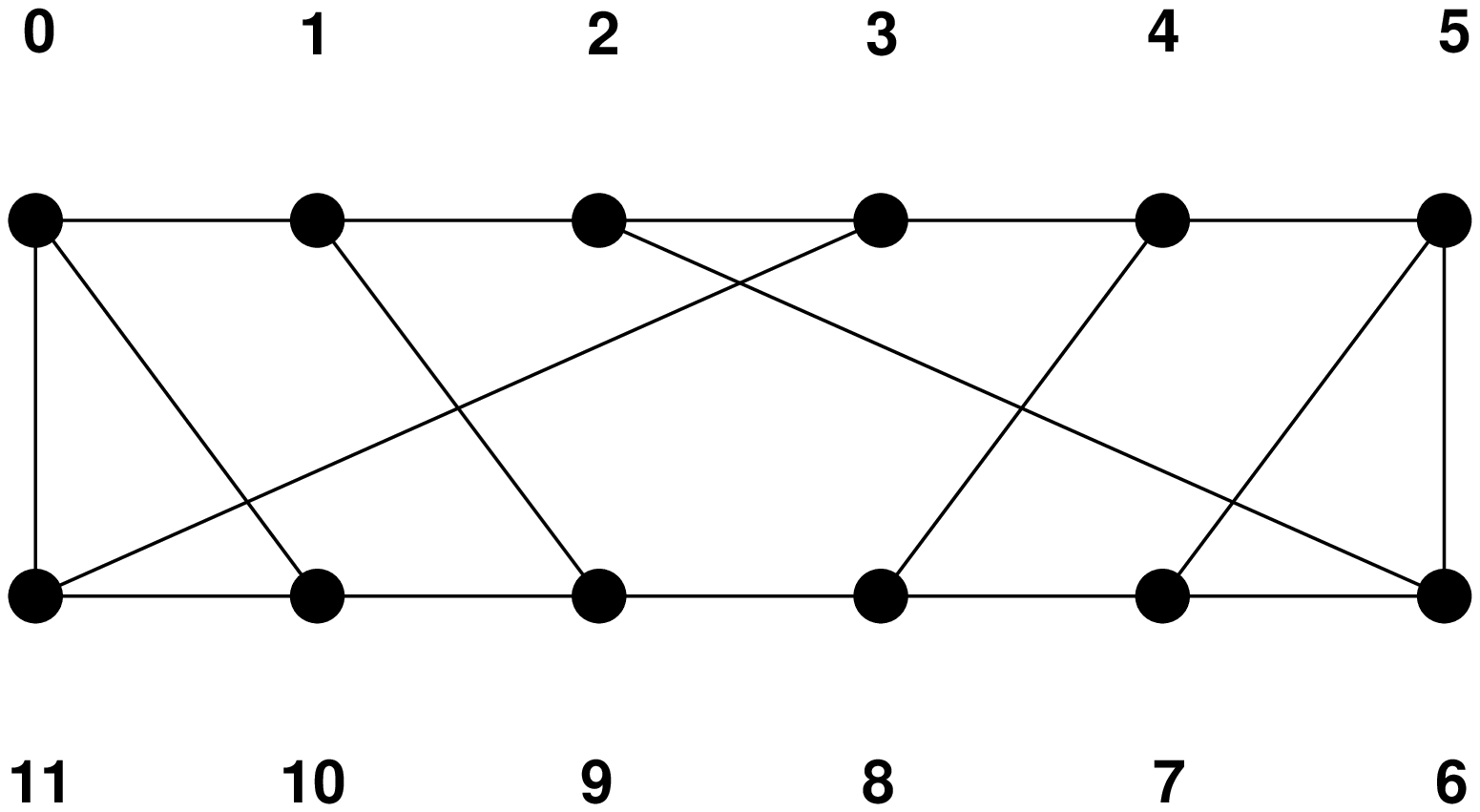}
\end{center}
\caption{The extremal graphs for $2n=12$ with the smallest number of cycles excluding $H_{12}$}
\label{Extremal 12}
\end{figure}
\begin{figure}[H]
\begin{center}
\includegraphics[width=0.42\textwidth]{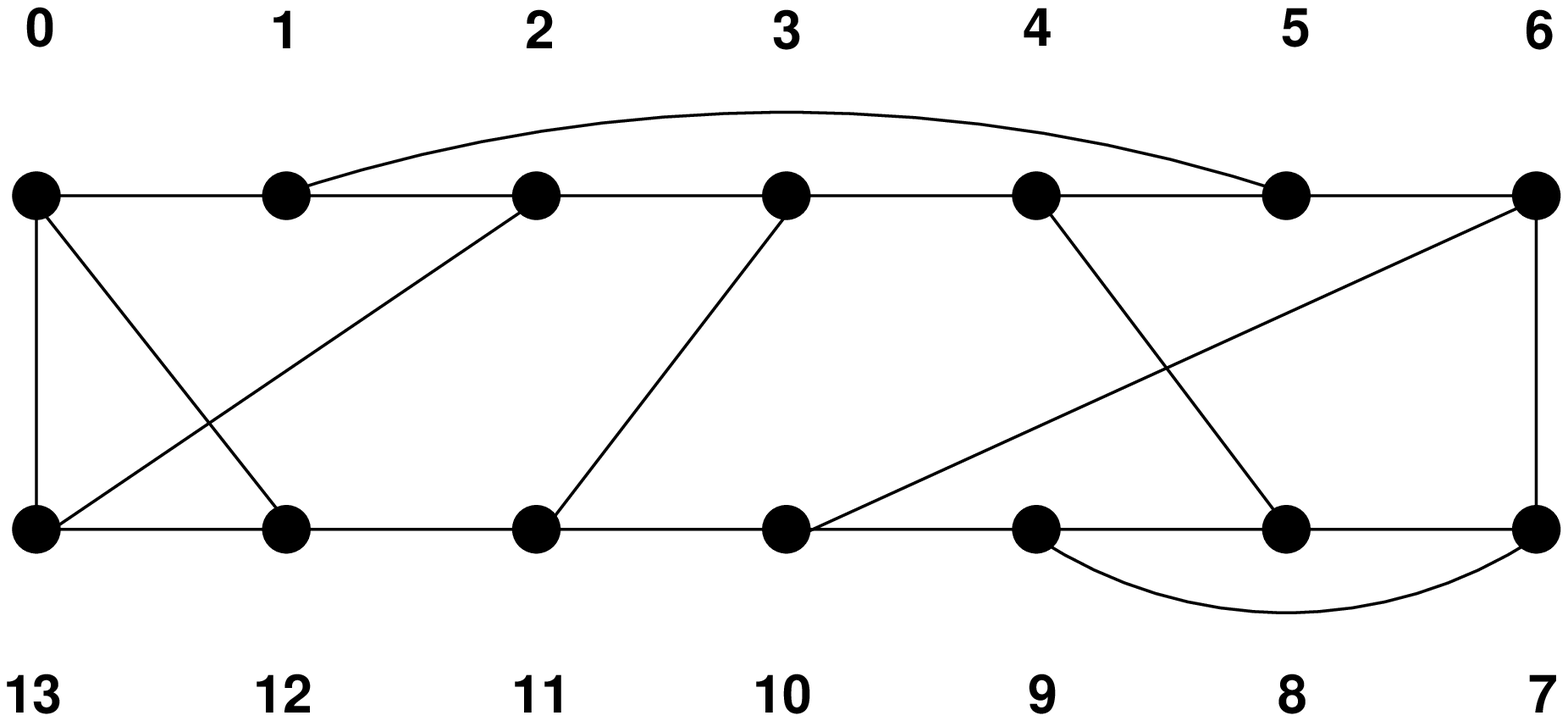}
\hspace{1 cm}
\includegraphics[width=0.42\textwidth]{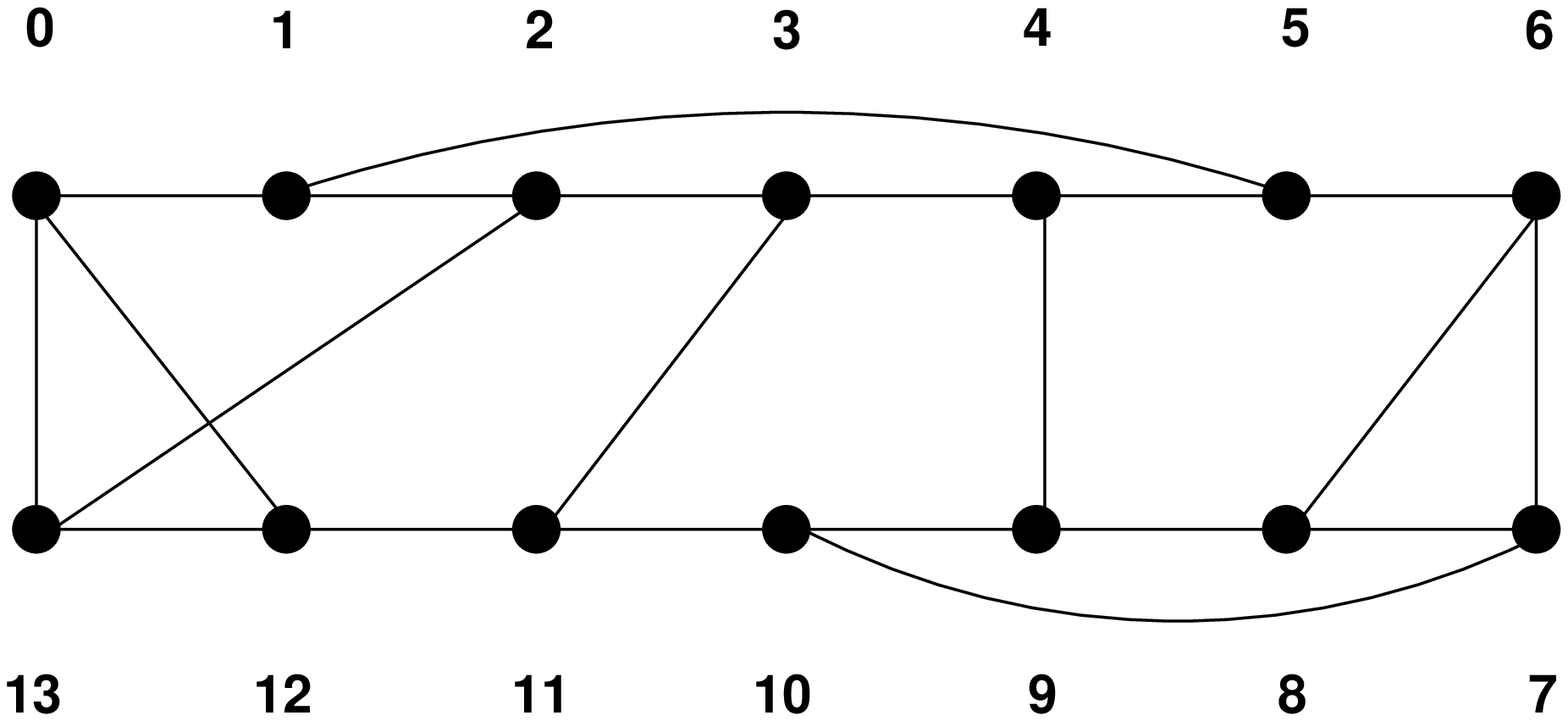}

\bigskip
\bigskip
\includegraphics[width=0.42\textwidth]{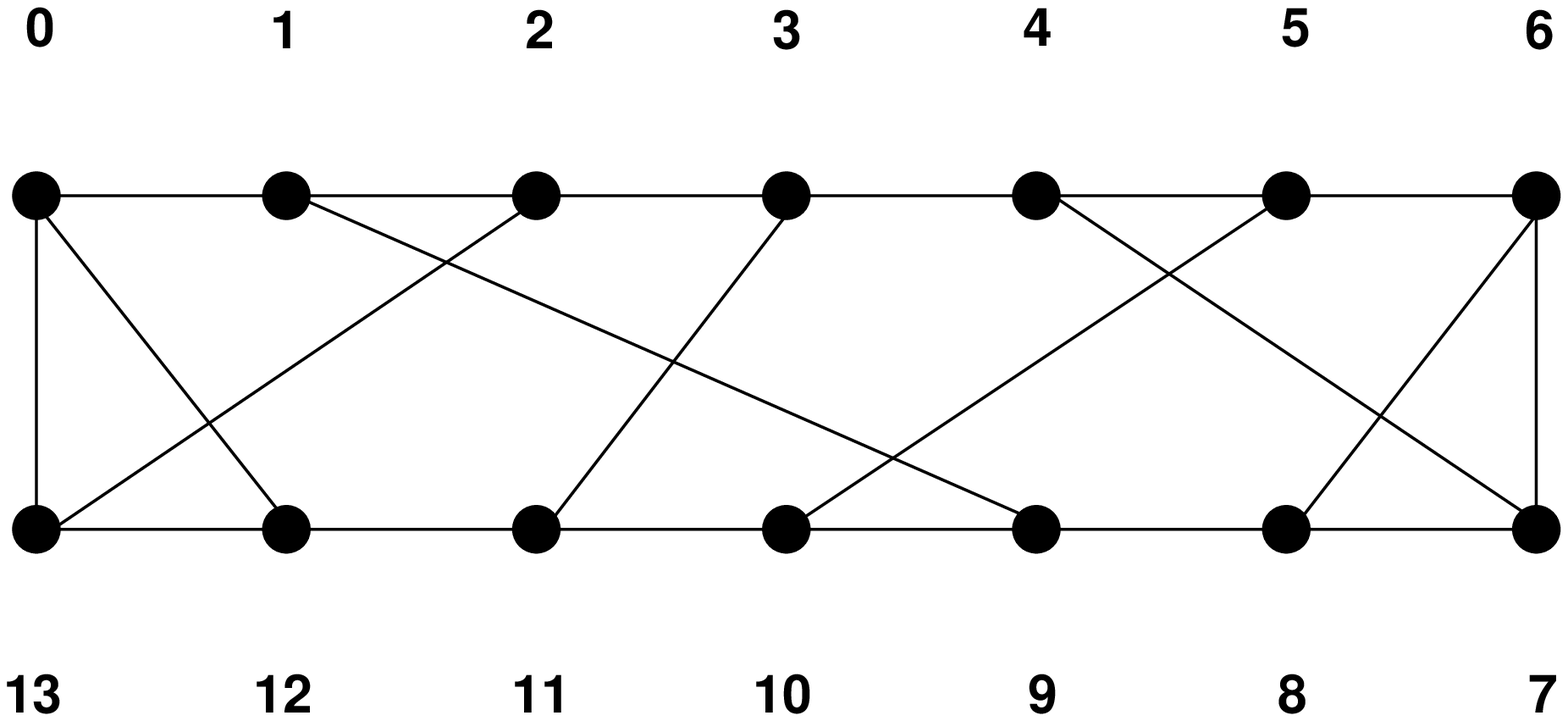}
\hspace{1 cm}
\includegraphics[width=0.42\textwidth]{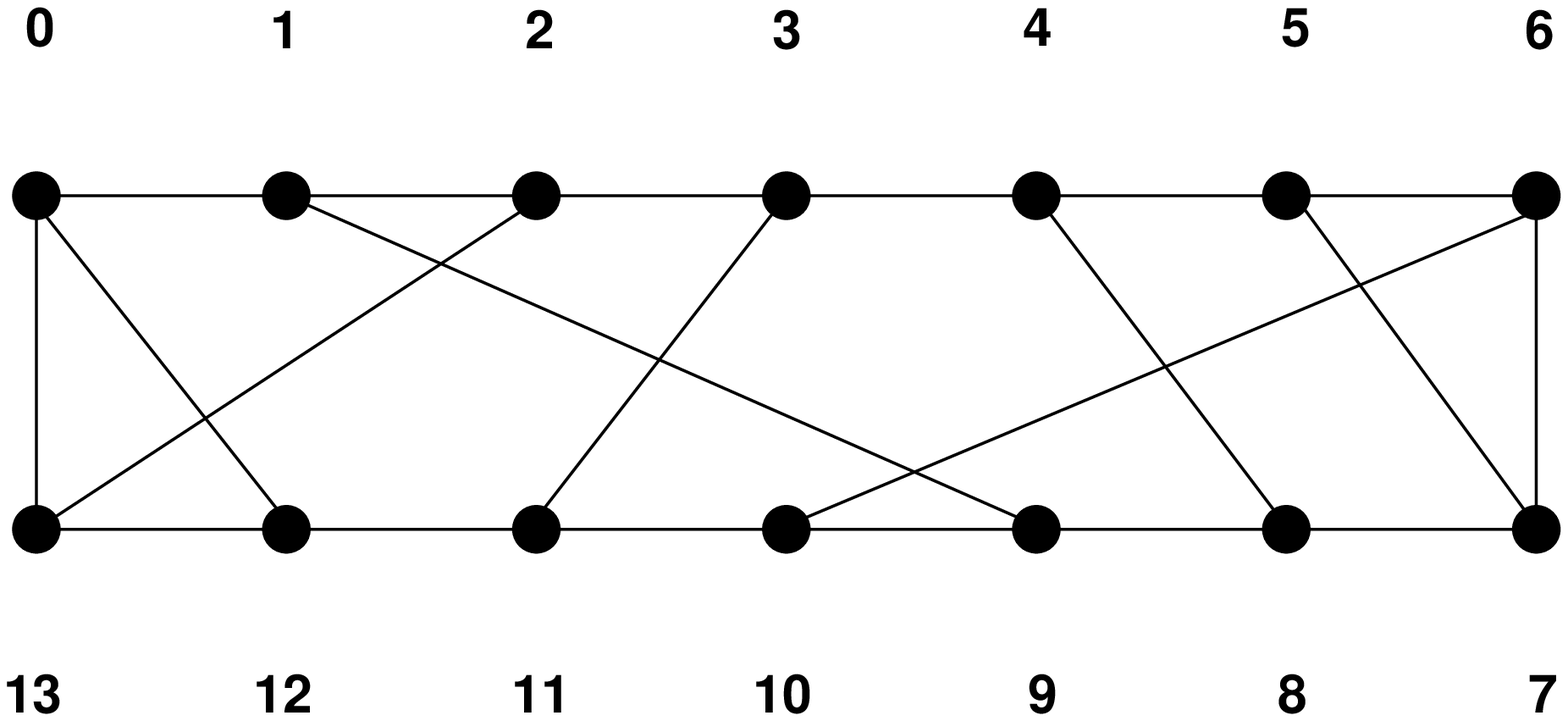}

\bigskip
\bigskip

\includegraphics[width=0.42\textwidth]{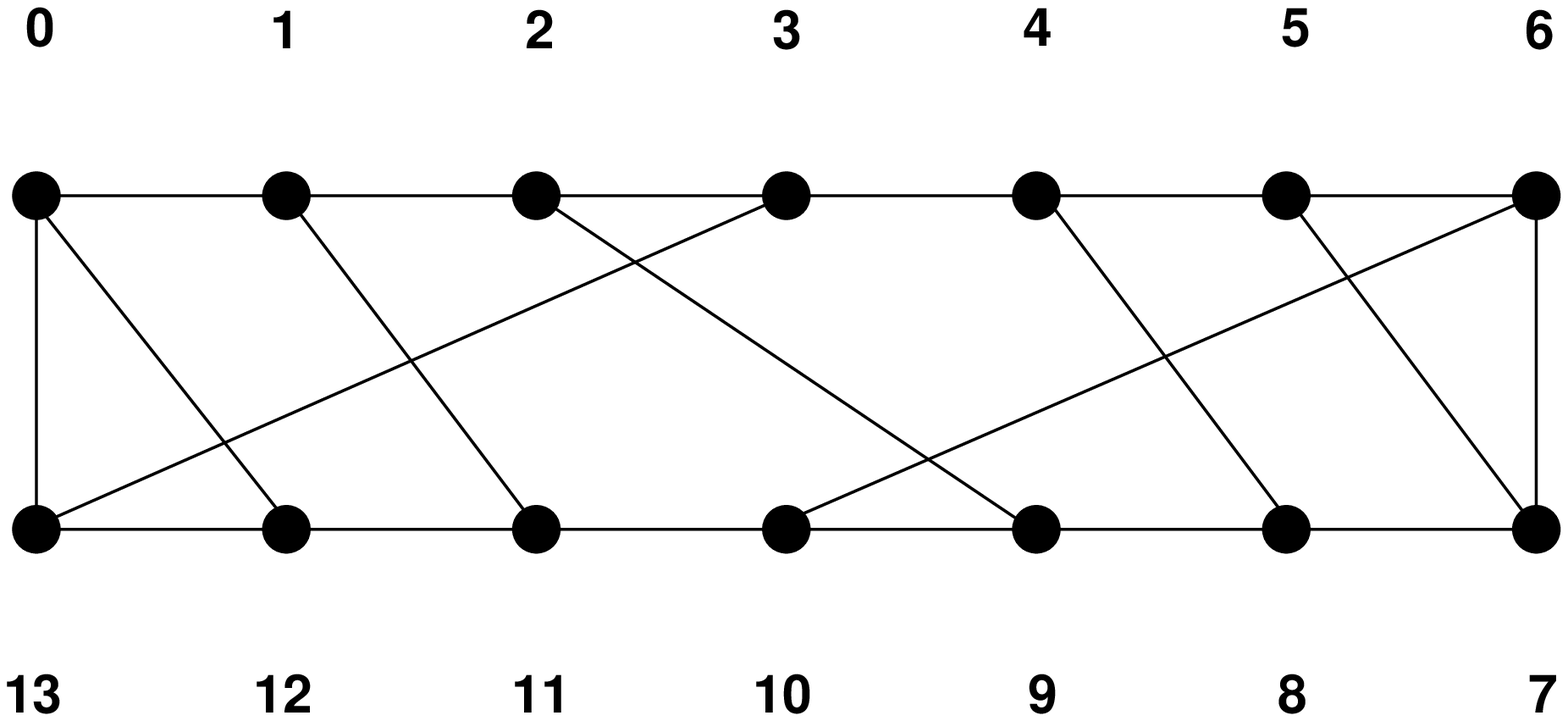}
\hspace{1 cm}
\includegraphics[width=0.42\textwidth]{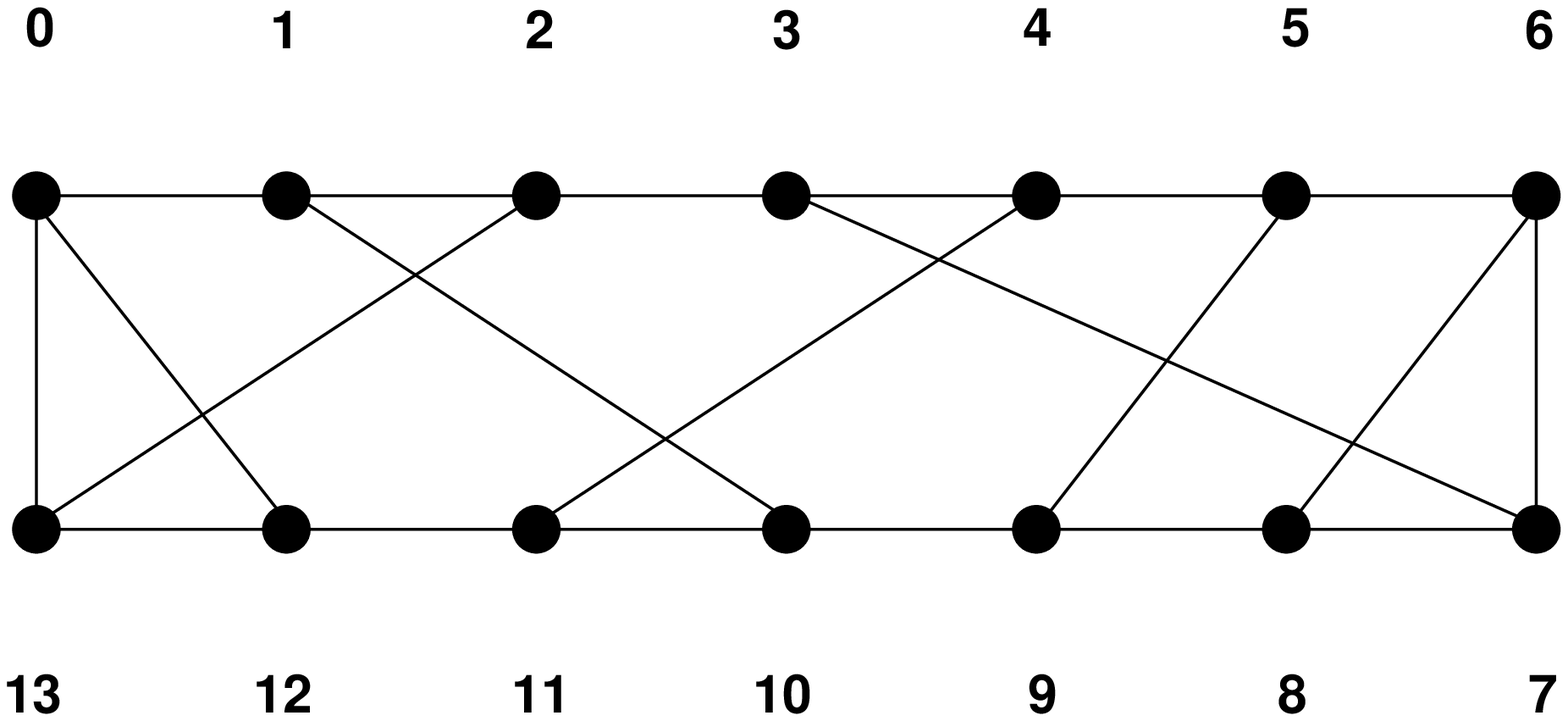}
\end{center}
\caption{The extremal graphs for $2n=14$ with the smallest number of cycles excluding $H_{14}$}
\label{Extremal 14}
\end{figure}
\section{Concluding Remarks}

So far there is no viable conjecture as to the largest number of cycles in
cubic graphs. Guichard~\cite{G} found $T(n)$ by an extensive computer
search for all $n\leq 18.$ Unfortunately, it is not clear from these results
what is the structure of the extremal graph.

Unlike the case of the largest number of cycles, this paper conjectures a structure
$H_{2n}$ for the smallest number of cycles in 3-connected hamiltonian
graphs.
The number of cycles in $H_{2n}$ is derived, and
the conjecture is verified using extensive computer searches for up
to $2n = 16$. The paper also presents a proof technique that could be refined
to improve the known upper bound on the largest number of cycles in a
hamiltonian graph. %

Extensive computer searches shall be carried on in future to
verify the conjecture for~$2n > 16$.
The searches will also find all the
graphs exhibiting the largest number of cycles. Hopefully, this would result
in identifying common extremal graph structures across different graph
sizes. % for the largest $n$ that could practically be handled.
Investigating these structure, if any, could lead to new venues on how to
determine the largest number of cycles for this class of graphs.


\begin{thebibliography}{99}

\bibitem{A}
\uppercase{Ahrens, W.}:
\textit{\"{U}ber das Gleichungssystem einer Kirchhoffschen Galvanischen Stromverzweigung}, Math. Ann. \textbf{49} (1897), 311--324.

\bibitem{Aldred}
\uppercase{Aldred, R. E. L.---Thomassen, C.}:
\textit{On the Number of Cycles in 3-Connected Cubic Graphs},
J. Combinatorial Th. Ser. B \textbf{71} (1997), 79--84.

\bibitem {Aldred1}
\uppercase{Aldred, R. E. L.---Thomassen, C.}:
\textit{On the Maximum Number of Cycles in a Planar Graph},
J. Graph Th. \textbf{57} (2008), 255--264.

\bibitem{Clark}
\uppercase{Barefoot, C. A.---Clark, L.---Entringer, R.}:
\textit{Cubic Graphs with the Minimum Degree of Cycles},
Congressus Numerantium \textbf{53} (1986), 49--62.

\bibitem{Leon}
\uppercase{Brillouin, L.}:
\textit{Science and Information Theory},
$2^{nd}$ Edition, New York, Academic Press, 1962.

\bibitem{Ent}
\uppercase{Entringer, R. C.---Slater, P. J.}:
\textit{On the Maximum Number of Cycles in a Graph},
Ars Combinatoria \textbf{11} (1981), 289--294.

\bibitem{G}
\uppercase{Guichard, D. R.}:
\textit{The Maximum Number of Cycles in Graphs},
Congressus Numerantium \textbf{121} (1996), 211--215.

\bibitem{H}
\uppercase{Horak, P.}:
\textit{On Graph with Many Cycles},
Discrete Math. \textbf{331} (2014), 1--2.

\bibitem{Anton}
\uppercase{Leykin, A.---Verschelde, J.---Zhao, A.}:
\textit{Evaluation of Jacobian Matrices for Newton’s Method with Deflation to
Approximate Isolated Singular Solutions of Polynomial Systems},
In Symbolic-Numeric Computation,
Birkhäuser Basel, 2007, Ch. 16, pp. 269--278.

\bibitem{mccabe1}
\uppercase{McCabe, T. J.}:
\textit{A Complexity Measure},
IEEE Transactions on Software Engineering \textbf{4} (1976), 308--320.

\bibitem{mccabe2}
\uppercase{T. J. McCabe, T. J.---Butler, C. W.}:
\textit{Design complexity measurement and testing},
Communications of the ACM \textbf{32} (1989), 1415--1425.

\bibitem{R}
\uppercase{Rautenbach, D.---Stella, I.}:
\textit{On the Maximum Number of Cycles in a Hamiltonian Graph},
Discrete Math. \textbf{304} (2005), 101--107.

\bibitem{watson}
\uppercase{Watson, A. H.---McCabe, T. J.---Wallace, D. R.}:
\textit{Structured Testing: A Testing Methodology Using the Cyclomatic Complexity Metric},
NIST special Publication \textbf{500} (1996), 1--114.

\end{thebibliography}
\end{document}